\pdfoutput=1
\documentclass[acmsmall, nonacm]{acmart}

\usepackage{balance}

\usepackage{amsmath}
\usepackage{amsthm}
\usepackage{enumerate}
\usepackage{url}
\usepackage{amsfonts}
\usepackage{hyperref}
\usepackage{cleveref}

\theoremstyle{definition}
\newtheorem{definition}{Definition}
\numberwithin{definition}{section} 
\newtheorem*{fact}{Fact}
\newtheorem{theorem}{Theorem}
\numberwithin{theorem}{section} 

\newtheorem{lemma}[theorem]{Lemma}

\newtheorem{corollary}[theorem]{Corollary}

\title{Regulation of Algorithmic Collusion}

\author{Jason D. Hartline}

\author{Sheng Long}

\author{Chenhao Zhang}
\email{{hartline,shenglong,chenhao.zhang.rea}@u.northwestern.edu}
\affiliation{%
  \institution{Northwestern University}
  \city{Evanston}
  \country{USA}
}

\usepackage{tikz}
\usepackage{pgfplots}
\usepackage{wrapfig}
\usepackage{subcaption}
\usetikzlibrary{intersections, calc}
\usetikzlibrary{shapes, arrows.meta}
\usetikzlibrary{overlay-beamer-styles}
\usetikzlibrary{patterns}
\usetikzlibrary{fit}
\usepgfplotslibrary{fillbetween}

\newcommand{\setsize}[1]{{\left|#1\right|}}

\DeclareMathOperator{\argmax}{argmax}
\DeclareMathOperator{\argmin}{argmin}

\newcommand{\given}{\,\mid\,}

\newcommand{\prob}[2][]{\text{\bf Pr}\ifthenelse{\not\equal{}{#1}}{_{#1}}{}\!\left[{\def\givenn{\middle|}#2}\right]}
\newcommand{\expect}[2][]{\text{\bf E}\ifthenelse{\not\equal{}{#1}}{_{#1}}{}\!\left[{\def\givenn{\middle|}#2}\right]}

\newcommand{\tparen}{\big}
\newcommand{\tprob}[2][]{\text{\bf Pr}\ifthenelse{\not\equal{}{#1}}{_{#1}}{}\tparen[{\def\given{\tparen|}#2}\tparen]}
\newcommand{\texpect}[2][]{\text{\bf E}\ifthenelse{\not\equal{}{#1}}{_{#1}}{}\tparen[{\def\given{\tparen|}#2}\tparen]}

\newcommand{\sprob}[2][]{\text{\bf Pr}\ifthenelse{\not\equal{}{#1}}{_{#1}}{}[#2]}
\newcommand{\sexpect}[2][]{\text{\bf E}\ifthenelse{\not\equal{}{#1}}{_{#1}}{}[#2]}

\newcommand{\priceset}{\mathcal{P}}
\newcommand{\regret}{R}
\newcommand{\ExpRegret}{ER} 

\newcommand{\jointpricedist}{\Pi}

\newcommand{\distover}[1]{\Delta(#1)}

\newcommand{\EExp}[2]{
\mathbb{E}_{#1}[#2]
}

\newcommand{\Er}[2]{ER^T(#1,#2)}

\newcommand{\EstEr}[2]{\widetilde{ER}^T(#1,#2)}

\newcommand{\EstNR}{\widetilde{NR}}

\usepackage{accents}
\newcommand{\ubar}[1]{\underaccent{\bar}{#1}}

\newcommand{\costrange}{[\ubar{c},\bar{c}]}

\newcommand{\estplcost}{\tilde{c}}
\newcommand{\plcost}{c_{\ast}}

\newcommand{\pmax}{\overline{p}}
\newcommand{\pimin}{\underline{\pi}}

\newcommand{\Rthresh}{\overline{r}}
\newcommand{\Tthresh}{\overline{T}}

\newcommand{\TrueEr}{R^T}
\newcommand{\PlEr}{R^T_{\ast}}
\newcommand{\EstPlEr}{\widetilde{R}^T}
\newcommand{\EstUCB}{\operatorname{UCB}^T}

\begin{document}
\acmYear{2024}\copyrightyear{2024}
\acmConference[CSLAW '24]{Symposium on Computer Science and Law}{March 12--13, 2024}{Boston, MA, USA}
\acmBooktitle{Symposium on Computer Science and Law (CSLAW '24), March 12--13, 2024, Boston, MA, USA}
\acmDOI{10.1145/3614407.3643706}
\acmISBN{979-8-4007-0333-1/24/03}

\begin{abstract}
First version: January, 2024\\
This version: August, 2024
\vspace{1em}

\noindent Consider sellers in a competitive market that use algorithms to adapt their prices from data that they collect.  In such a context it is plausible that algorithms could arrive at prices that are higher than the competitive prices and this may benefit sellers at the expense of consumers (i.e., the buyers in the market).  This paper gives a definition of plausible algorithmic non-collusion for pricing algorithms.  The definition allows a regulator to empirically audit algorithms by applying a statistical test to the data that they collect.  Algorithms that are good, i.e., approximately optimize prices to market conditions, can be augmented to collect the data sufficient to pass the audit.  Algorithms that have colluded on, e.g., supra-competitive prices cannot pass the audit. The definition allows sellers to possess useful side information that may be correlated with supply and demand and could affect the prices used by good algorithms. The paper provides an analysis of the statistical complexity of such an audit, i.e., how much data is sufficient for the test of non-collusion to be accurate.
\end{abstract}
\maketitle

\section{Introduction}

Algorithms are increasingly being used to price goods and services in competitive markets.  Several recent papers have shown that in certain settings, some configurations of certain pricing algorithms can find and maintain supra-competitive prices when in competition with each other \citep{calvano2020artificial, banchio2022artificial,asker2023impact, assad2020algorithmic}. As a result, (a) regulators may be concerned about how the risk of algorithmic collusion can be mitigated and the consistency of this regulation with legal standards for collusion, (b) individual sellers may be interested in algorithms that provably do not collude, and (c) third-party platforms. like AirBnB and eBay, may want to only recommend pricing algorithms to their sellers that will not risk incriminating the platforms themselves for price fixing \citep{harrington2022effect}.
Several papers have proposed 
 ways to change competition and antitrust law in response to the potential risks brought by algorithmic collusion \citep{beneke2019artificial, gal2023limiting, harrington2018developing}. This paper complements these proposals with a method for regulating algorithmic collusion from data.

For individual adoption of pricing algorithms, our test for algorithmic collusion parallels the role that overt communication plays in the modern legal theory of (non-algorithmic) collusion.
Under current, if controversial, understandings of American antitrust law,  an \textit{express} agreement (in the form of overt communication) is a prerequisite to establishing liability under the Sherman Act.\footnote{The act of colluding without an express agreement is known as ``tacit collusion''. The Supreme Court defined tacit collusion as ``the process, \textit{not in itself unlawful}, by which firms in a concentrated market might in effect share monopoly power, setting their prices at a profit-maximizing, supra-competitive level by recognizing their shared economic interests and their interdependence with respect to price and output decisions'' \citep{1993brooke}.} 
Courts cannot read the minds of the sellers to understand their pricing strategies and therefore prefer to rely on evidence of overt illegal coordination \cite{harrington2018developing}. Similarly, in the setting of algorithmic collusion, we might not know some of the fundamentals that guide a seller's pricing algorithm (e.g., seller's costs or information about the demand), but we can determine whether an outcome is competitive for some fundamentals. We refer to such outcomes as \textit{plausible non-collusion}. What is left out is outcomes that are non-competitive for any fundamentals, i.e., regardless of what is in the minds of the sellers.  We argue subsequently that there is no loss using our test for algorithmic collusion to forbid algorithms that obtain such non-competitive outcomes.

For third-party platforms of online marketplaces or algorithm vendors who are recommending or selling pricing algorithms to sellers in a market, the coordination of a third party on an algorithm that is known to obtain supra-competitive prices is illegal by current standards \citep{harrington2022effect}.  Our test for algorithmic collusion identifies such problematic algorithms.

The possibility of algorithmic collusion creates greater risk of supra-competitive prices.  Thus, regulators and lawmakers may desire methods for establishing algorithmic collusion beyond the legal standards of non-algorithmic collusion.    \citet{harrington2018developing} noted that algorithms afford introspection that non-algorithmic human agents do not afford.  He discussed regulating algorithmic collusion by prohibiting pricing algorithms with certain properties and proposed a few ``inside the head'' approaches to check if an algorithm has any of these properties. One approach is for the regulator to check the source code of the algorithms. This approach has several drawbacks. On one hand, this approach seems to require costly and detailed scrutiny by experts \citep{krollAccountableAlgorithms2017} and has the potential to leak the intellectual property of the algorithm developers \citep{RuckelshausMonsantoCo1984}. On the other hand, the source code of popular black box algorithms such as those based on deep neural networks gives little information about its behavior.

Another approach \citet{harrington2018developing} discussed is to conduct dynamic testing on the pricing algorithms, i.e., running the algorithms by feeding them with simulated inputs to observe their behavior. Although dynamic testing is generally considered an effective approach for detecting software bugs, there are still challenges when applied to understanding the behaviors of pricing algorithms. To make better pricing decisions in vibrant market environments, the inputs to pricing algorithms are usually large in dimensions and dynamic. It is infeasible to exhaust all or even a small portion of the possible inputs the algorithms could take. Further, the inputs the algorithms receive from the environments in which they are deployed could be very different from the simulated ones the regulators could expect. Even worse, the not-so-recent Volkswagen emissions scandal and the emergence of the field of non-adversarially robust machine learning \citep{carlini2017adversarial,biggio2013evasion} showed that inputs, where collusion happens, could be adversarially hidden to evade scrutiny. As the early pioneer of computer science Edsger W. Dijkstra noted regarding testing for bug-finding, ``The first moral of the story is that program testing can be used very effectively to show the presence of bugs but never to show their absence'' \citep{dijkstra1970reliability}. Dynamic testing can be largely uninformative for the behavior of pricing algorithms on inputs that are not tested during simulation but show up during the algorithms' actual deployment. Similar points for dynamic testing have also been made by \citet{desai2017trust}. 

This paper takes a different approach.  It identifies an empirical condition that can be checked from data logged by the algorithm while deployed to prove statistically that the algorithm is not colluding under reasonable assumptions. It provides a way to augment any ``good'' algorithm to collect this data without significantly harming its performance. ``Bad'' algorithms that collude on implausibly competitive prices cannot be augmented to collect data and pass the audit. 
Our framework enables algorithms to prove that they are plausibly non-collusive and opens the opportunity for new legal standards for enforcing non-collusion, namely, requiring the algorithms used to continually pass such a test.

The paper develops an empirical definition of plausible non-collusion that has the following two groups of properties: 
\begin{itemize}
    \item Economic properties:
    \begin{itemize}
        \item (unilateral) non-collusion is a unilateral property that an algorithm can satisfy independently of what other algorithms are doing.
        \item (information compatible) it allows the sellers to use side information that may be correlated.
        \item (optimal)  optimizing is not collusion.
    \end{itemize}
    \item Legal properties:
    \begin{itemize}
        \item (plausibly correct) algorithms that collude on supra-competitive price inconsistent with plausible preferences and beliefs of sellers can not satisfy it.
        \item (minimum burden of compliance) There are known good algorithms, i.e., those that do not use suboptimal prices, that satisfy the definition. Any new good algorithm can be augmented to collect the necessary data to satisfy the definition with little performance loss. 
    \end{itemize}
\end{itemize}

 Justification of economic properties is as follows. It is critical for a definition of non-collusion to have the property of being unilateral. A seller should always be able to adopt a pricing strategy that is non-collusive, regardless of what other sellers do. It may surprise the reader that our definition of non-collusion allows correlation of the the sellers' behavior. This correlation of behavior is required to handle side information that could be correlated. Consider the following example, the demand for hotels in the business district is higher during the week and the day-of-week is known to all sellers. Setting different prices in response to known differences in demand is not collusion. Last but not least, if the seller is optimizing given the information they obtained, then this act of optimization is not collusion. 
We argue that if algorithms satisfy the three economic properties, then they are not colluding. On the other hand, if algorithms do not satisfy these three properties, then something undesirable is happening that regulators may wish to rule out.

The legal properties our definition satisfies make it appropriate for regulators to require it for pricing algorithms. Our definition of non-collusion rules out the algorithms that collude on supra-competitive prices identifiable without ``seeing through the minds'' of the sellers deploying them. Note that this definition does leave two ways sellers could be supra-competitive but plausibly competitive: by acting as though their costs are higher than their actual costs or ignoring information that they may have in the market that would result in lower prices. Since the regulator cannot see into the minds of the sellers, the legal standards suggest that such plausible non-collusion is not illegal.\footnote{There may be opportunities to be further restrictive for pricing algorithms recommended by platforms like AirBnB and eBay that must be configured by individual sellers with information such as their costs.  When the costs are reported and can be logged by the algorithm, it could be required that the prices are competitive for the reported costs.} This parallels the modern legal theory of regulating non-algorithmic collusion via explicit agreements. 
Our definition also places a minimum burden on sellers deploying pricing algorithms that satisfy the economic properties. Algorithms that satisfy the properties and collect the relevant data are known and new algorithms satisfying the property can be augmented to collect relevant data with minor effects on their performance. Therefore, it is feasible for the regulator to make it a requirement for all pricing algorithms without putting excessive burdens on firms adopting them.

In summary, our main contributions are as follows: (1) a definition of non-collusion; (2) a framework for empirically auditing pricing algorithms for whether they satisfy plausibly non-collusion; and (3) an instantiation of our framework and an analysis of its statistical complexity. Using our provided framework, algorithms can collect data to prove their plausible non-collusion and regulators can audit algorithms without checking source code or limiting algorithms to be a pre-approved set.
 The main technical analysis of our definition of non-collusion contributes a quantification of the sample complexity (i.e., how much data is necessary) for an algorithm to collect to prove with high confidence that it is plausibly not colluding.  

\subsection{Related literature}

 As of current, collusion is regulated in the US legal system by three core federal antitrust laws: the Sherman Act (1890), the Federal Trade Commission Act (1914), and the Clayton Act (1914). The standard legal definition for collusion leaves open the issue of whether an express agreement through overt communication is needed for the behavior to be deemed illegal \citep{chassangRegulatingCollusion2023}. Earlier rulings such as  \citet{InterstateCircuit1939} and \citet{AmericanTobaccoCo1946} found firms engaging in illegal collusion without any explicit agreement via communication. However, more recent judicial decisions, such as \citet{1993brooke}, have evolved to require the presence of such agreements. Tacit collusion by itself is \textit{not} a violation of the Sherman Act \citep{InreTextMessagingAntitrustLitigation2015}. The raison d'\^etre for requiring express agreement is that it gives an explicit condition that courts can establish. Courts have declined to impose antitrust liability for tacit collusion alone, partly because it is difficult to distinguish tacit collusion from independent decision-making that simply takes into account the actions of rivals in oligopolistic markets \citep{yao1993game}. As Judge Breyer puts it, ``[it] is not because such [parallel] pricing is desirable (it is not), but because it is close to impossible to devise a judicially enforceable remedy for `interdependent pricing'. How does one order a firm to set its prices \textit{without regard} to the likely reactions of its competitors?'' \citep{1988clamp} The courts use an analytical framework that permits an inference of conspiracy where there is circumstantial evidence of tacit collusion ``plus'' something else that tends to ``exclude the possibility that the alleged conspirators acted independently'' \citep{yao1993game, kovacic2011plus}. 
\citet{kovacic2000antitrust} provides a detailed review of the evolution of thinking about competition as reflected by major antitrust decisions and research in industrial organizations.

Economists study collusion mostly via the lens of oligopoly theory. Non-cooperative game theory is the currently accepted economic model to analyze oligopolistic interactions \citep{yao1993game}. Despite the vast literature, there is no unified theory of oligopolistic rivalry, though the mainstream models share common assumptions and approaches. That is, economists agree on what elements a ``good'' model should contain \citep{yao1993game, werden2004economic}. Earlier works on oligopoly theory include \citet{stigler1964theory} and \citet{friedman1971non}. \citet{werden2004economic} provides a good review of basic terms and concepts in game theory as well as modern oligopoly theory.  Our definition takes a similar approach, but relaxes competitive behavior to be unilateral and to allow correlated side information.

There has also been a lot of recent work on pricing algorithms and whether/how they could lead to potentially collusive outcomes. Empirical work, such as \citet{assad2020algorithmic}, studied the effects of pricing algorithms in the German retail gasoline market. They found that prices increased substantially after both firms in a duopoly switched from manual to algorithmic pricing \cite{gal2023limiting}. In a well-cited simulation study, \citet{calvano2020artificial} showed that a commonly used reinforcement learning algorithm learned to initiate and sustain a supra-competitive equilibrium when only instructed to maximize its own profits in a simultaneous, repeated price competition. \citet{klein2021autonomous} observes a similar reward-punishment pattern as \citet{calvano2020artificial}. At the same time, there is also research that provides evidence for the opposite argument. For example, \citet{abada2023artificial} showed that seemingly collusive outcomes could originate in imperfect exploration rather than excessive algorithmic sophistication. \citet{den2022artificial} examined the Q-learning algorithm used in \citet{calvano2020artificial} in detail and concluded that ``simulations presented by Calvano et al. (2020a) do not give sufficient evidence for the claim that these types of Q-learning algorithms systematically learn collusive strategies.'' \citet{banchio2023adaptive} developed a theory of explaining the collusive behavior of learning algorithms by their statistical linkage.

Between the economic and legal literature, there seems to be a gap between how they view collusion. While the law examines whether competitors have taken possibly avoidable actions from which an anti-competitive agreement may be inferred, economic theory is more concerned with what final coordinated outcomes may be produced by certain conduct \citep{yao1993game}. This may explain why some legal scholars tend to use the term ``collusion'' more narrowly to refer to illegal cartelization only (and not legal oligopolistic coordination) \citep{gal2023limiting}. 
The term ``algorithmic collusion'' lends itself to different interpretations, and we use it throughout the paper to refer to ``algorithmic tacit collusion'' as opposed to ``algorithmic explicit collusion'',  where algorithms implement an existing collusive strategy potentially defined or agreed upon by humans \citep{gautier2020ai}.

Last but not least, this paper builds on an extensive literature on algorithms for dynamic learning of prices. Early papers by \citet{bar2002incentive,kleinberg2003value, blum2003online,blum2005near} show that the dynamic learning of prices fits into the framework of multi-armed bandit learning, enabling a large portfolio of well-studied algorithms to be successfully applied. Multi-armed bandit learning can be applied in repeated interactions with partial feedback\footnote{Here ``partial feedback'' is that the learner only learns the outcome of the selected price, they do not learn counterfactual outcomes of other prices.}. There is a canonical reduction from multi-armed bandit learning to online learning (with full feedback, e.g., where the learner also learns the payoffs of counterfactual prices) that employs propensity scoring, i.e., constructing unbiased estimators of counterfactual payoffs. \citet{blum2007external} reduce best-in-hindsight learning (a.k.a., external regret) to calibrated learning (a.k.a., internal or swap regret). \citet{Nekipelov2015} consider inferring values of low regret ad buyers from bidding data assuming a full feedback model. Our analysis is based on their definition of the rationalizable set of values and regrets for a buyer, naturally applied to the dual problem of a seller with a cost, and generalized from the full feedback setting to the partial feedback setting. 

Concurrent to our work, \citet{chassangRegulatingCollusion2023} informally discussed the idea of enforcing a property (known as ``no regret'') on pricing algorithms based on the observation of \citet{chassang2022robust}.  \citet{chassang2022robust} proposed a test of competitive behaviors in procurement auctions and applied it to real-world data. They define competitive behavior as perfect public
Bayesian equilibrium \citep{athey2008collusion} in Markov perfect strategies \citep{maskin2001markov}. They derive necessary conditions for beliefs of the firms participating procurement auctions to be consistent with competitive behaviors, which are testable with  data containing biding history of all firms. In their model, the cost of the firms  (sellers) can be different across rounds while we consider the setting where the firm has fixed cost. On the other hand, our approach is unilateral. We do not make any assumption on the belief of the firm of concern and do not require full information about the other participants of the market.
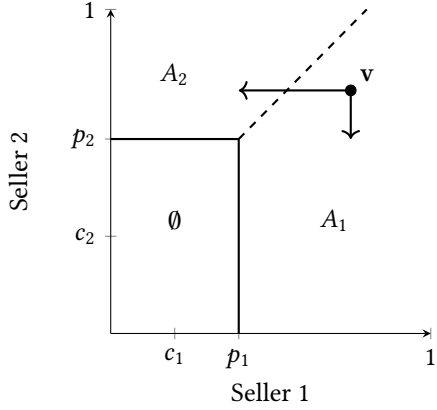
\begin{figure}
    \centering
    \begin{tikzpicture}
    \begin{axis}[
      width=2.3in,
      height=2.3in,
        axis lines = left,
        xlabel = {Seller 1},
        ylabel = {Seller 2},
        xmin=0, xmax=1,
        ymin=0, ymax=1,
        xtick={0.2,0.4,1},
        ytick={0.3,0.6,1},
        xticklabels={$c_1$,$p_1$,1},
        yticklabels={$c_2$,$p_2$,1},
    ]

    \addplot[black, thick] coordinates {(0.4,0) (0.4,0.6)};
    \addplot[black, thick] coordinates {(0,0.6) (0.4,0.6)};
    
    \addplot[black, thick, dashed] coordinates {(0.4,0.6) (0.8,1)};

    \node at (axis cs:0.7,0.35) {$A_1$};
    \node at (axis cs:0.2,0.8) {$A_2$};
    \node at (axis cs:0.2,0.35) {$\emptyset$};

      \addplot [mark=*,black] coordinates {(0.75,.75)} node[above right] {$\mathbf{v}$};
    
      \draw [black,thick,arrows=->] (axis cs: 0.75,.75) -- (axis cs: 0.75,0.6);
      \draw [black,thick,arrows=->] (axis cs: 0.75,.75) -- (axis cs: 0.4,0.75);

    \end{axis}
    \end{tikzpicture}
        \caption{Geometric illustration of the example at a particular round. Two sellers post prices $p_1$ and $p_2$ respectively. The buyer buys from Seller 1 if his valuation $\text{v} \equiv (v_1, v_2)$ (represented by the black dot), lies in the region $A_1$. He buys from Seller 2 if $\text{v}$ lies in the $A_2$. The buyer buys from neither seller if $\text{v}$ lies in the region $\emptyset$.}
        \label{fig:example}
\end{figure}

\section{Preliminaries}
\subsection{Dynamic Imperfect Price Competition}
\label{sec:dynamic_price_comp}
We consider a setting of dynamic imperfect price competitions with $k$ discrete price levels in which $n$ sellers repeatedly compete for selling one unit of good or service (hereafter referred to as ``good'') over $T$ rounds.  Seller $i$ has a fixed cost $c_{i}$ to produce a unit of the good. In each round $t$: 
\begin{itemize}
    \item Seller $i$ posts a price $p^t_i \in \priceset$, where $\priceset$ is the set of discretized price levels with $\setsize{\priceset}=k$.
    \item The market condition for seller $i$ is captured by a demand function\footnote{In mechanism design and auction contexts, demand functions are usually referred to as ``allocation rules''. This demand function encapsulates behaviors of consumers. } $x_i^t:\priceset^n\to [0,1]$. It produces the demand $x_i^t(p_i^t,p_{-i}^t)$ for seller $i$ where $p_{-i}^t=(p_1^t, \dots, p_{i-1}^t,p_{i+1}^t, \dots, p_n^t)$ is the prices of other sellers. In other words, $x_i^t$ is jointly determined by the prices posted by all sellers. Assuming normal goods, fixing $p_{-i}$, demand $x_i$ is monotonic in $p_i$.
    \item Seller $i$ gets a payoff of $u_i^t(p^t_i)=(p^t_i-c_{i})\,x_i^t(p^t_i,p^t_{-i})$. 
\end{itemize}

{An illustrative example with two sellers is shown in \autoref{fig:example}. Seller $1$ and $2$ have cost $c_{1},c_{2} \in [0,1]$ respectively. The price levels $\priceset \subseteq [0,1]$. At each round $t$, a buyer shows up with valuation $v^t_1$ and $v^t_2$ for the goods provided by the two sellers respectively. After seeing the prices $p^t_1$ and $ p^t_2$ posted by the two sellers, the buyer chooses to buy from seller $i$ that maximizes his utility $v^t_i-p^t_i$ if $v^t_i -p^t_i\geq 0$, breaking ties in favor of Seller 1.  He ``buys nothing'' if $v^t_i -p^t_i < 0$ for both $i=1$ and $2$. Suppose each buyer draws his valuations $(v^t_1, v^t_2)$ from the distribution $D^t$, the demand for Seller 1's good is
\begin{align*}
    x^t_1(p^t_1, p^t_2) &= \Pr_{(v^t_1, v^t_2)\sim D^t}[v^t_1-p^t_1 \geq \max(v^t_2-p^t_2,0)].
\end{align*}
A special case of this example is when the buyer's valuations of the goods of the two sellers are i.i.d. uniformly distributed over $[0,1]^2$ for every round, i.e., $D^t = U[0,1]\times U[0,1]$. If we further assume that each seller sets one fixed price to post for all rounds, with cost $c_1=0.1$ and $c_2=0.2$, the equilibrium prices are $p_1^{eq} \approx 0.50$ and $p_2^{eq} \approx 0.55$. However, if the two sellers collude by setting supra-competitive prices $p_1^c \approx 0.60$ and $p_2^c \approx 0.66$, they will get a higher total average revenue.

\subsection{Learning Problem of Sellers}
The dynamic pricing problem of each seller is essentially an online learning problem. At round $t$,  seller $i$'s pricing algorithm chooses a price distribution $\pi_i^t \in \distover{\priceset}$ from the set of all distributions over the prices based on her information about the history and the market. She then draws a price $p_i^t$ from the distribution $\pi_i^t$ and posts $p_i^t$.  Following the model of online learning with bandit feedback, we make the minimum assumption about the information a seller possesses: After posting $p_i^t$ at round $t$, seller $i$ observes the demand for her good $x_i^t(p_i^t,p_{-i}^t)$ and hence her payoff $u_i^t(p_i^t,p_{-i}^t)$  as she knows her cost $c_i$. 

To measure the performance of a seller's online learning algorithm, we employ the concept of hindsight \textit{calibrated regret}, which is defined as the benefit of the best-in-hindsight remapping of the prices chosen by the algorithms. We instantiate this definition in the setting of price competition. 

\begin{definition} Given a sequence of historical prices $\{(p_i^t,p_{-i}^t)\}_{t=1}^T$ and demand functions $\{x_i^t(\cdot)\}_{t=1}^T$ for seller $i$ with cost $c_i$, the \emph{hindsight (realized) regret against a fixed price remapping $\sigma : \priceset \to \priceset$} is
\begin{align*}
    \regret^T(\sigma,c_i) &= \frac{1}{T}\sum_{t=1}^T u_i^t(\sigma(p_i^t),p_{-i}^t) - u_i^t(p_i^t,p_{-i}^t) \\ &= \frac{1}{T}\sum_{t=1}^T(\sigma(p_i^t)-c_i)x_i^t(\sigma(p_i^t)) - (p_i^t-c_i)x_i^t(p_i^t).
\end{align*}
The \emph{hindsight calibrated (realized) regret} is defined as  the maximum hindsight calibrated regret over all remappings, $\max_\sigma \regret^T(\sigma, c_i)$.
\end{definition}

Since a seller's pricing algorithm chooses a distribution of prices at each round, a performance measure without considering a particular realization of the distributions is the expected regret.

\begin{definition}
    Given a sequence of historical price distributions $\{\pi_i^t\}_{t=1}^T$, prices $\{p_{-i}^t\}_{t=1}^T$, and demand functions $\{x_i^t(\cdot)\}_{t=1}^T$, the \emph{expected regret of seller $i$ with cost $c_i$ against a fixed price remapping $\sigma :\priceset \to \priceset$} is
\begin{align*}
    \ExpRegret^T(\sigma, c_i) &= \frac{1}{T}\sum_{t=1}^T\EExp{p_i^t \sim \pi_i^t}{u_i^t(\sigma(p_i^t),p_{-i}^t) - u_i^t(p_i^t,p_{-i}^t)} \\&= \frac{1}{T}\sum_{t=1}^T\EExp{p_i^t \sim \pi_i^t}{(\sigma(p_i^t)-c_i)\cdot x_i^t(\sigma(p_i^t),p_{-i}^t)\\&\phantom{\frac{1}{T}\sum_{t=1}^T\EExp{p_i^t \sim \pi_i^t}{}} - (p_i^t-c_i)\cdot x_i^t(p_i^t,p_{-i}^t)}.
\end{align*}
The \emph{expected calibrated regret} for $i$ is defined as $\max_\sigma \ExpRegret^T(\sigma, c_i)$.
\end{definition}

Note that calibrated regret is also called ``swap regret'' or ``internal regret'' in the literature. In addition, there is a common weaker notion of regret, known as the \textit{hindsight external regret}. Hindsight external regret is defined as the benefit of deviating to a single best-in-hindsight action.
\begin{definition}
\label{def:external-regret}
    Given a sequence of historical prices $\{(p_i^t,p_{-i}^t)\}_{t=1}^T$ and demand functions $\{x_i^t(\cdot)\}_{t=1}^T$ for seller $i$ with cost $c_i$, the \emph{hindsight external regret}
\begin{align*}
    R_{\mathrm{ext}}^t(c_i) &=  \max_p \frac{1}{T}\sum_{t=1}^T u_i^t(p,p_{-i}^t) - u_i^t(p_i^t,p_{-i}^t) \\ &= \max_p \frac{1}{T}\sum_{t=1}^T(p -c_i)x_i^t(p) - (p_i^t-c_i)x_i^t(p_i^t).
\end{align*}
\end{definition}
Unlike calibrated regret, the definition of hindsight external regret does not allow beneficial side information. Later in this paper, we argue that hindsight external regret is insufficient precluding collusion. 

Based on the results of \citet{auer2002nonstochastic}, 
\citet{blum2007external} and \citet{stoltz2005incomplete} give algorithms that achieve vanishing expected calibrated regret for an individual seller regardless of the market condition and other sellers' behavior. Such algorithms are among those generally referred to as ``no-regret learning algorithms''.

\begin{theorem}[\citealp{blum2007external,stoltz2005incomplete}]
    For a seller $i$ with cost $c_i$, there exists an online algorithm such that:
    \begin{itemize}
        \item At round $t$, it outputs a price distribution $\pi_i^t \in \priceset$ based on the history $\{p_i^{s}\}_{s=1}^{t-1}$ and $\{x_i^{s}(p_i^{s})\}_{s=1}^{t-1}$.
        \item The expected regret of seller $i$ given any sequence of $\{p_{-i}^t\}_{t=1}^T$ and $\{x_i^t(\cdot)\}_{t=1}^T$ with the algorithm's output $\{\pi_i^t\}_{t=1}^T$ satisfies  $\max_{\sigma}\ExpRegret^T(\sigma,c_i) = o(1)$.
    \end{itemize}
\end{theorem}

A characteristic of no-regret learning algorithms is that they lead to \emph{correlated equilibrium} \citep{foster1997calibrated}. 
Correlated equilibrium \citep{aumann1974subjectivity} is a static equilibrium concept that is often described as a mediator that draws a profile of prices from a joint distribution and privately suggests a corresponding price for each seller. The joint distribution of prices is a correlated equilibrium if each seller has no incentive to deviate from their suggested price.
\begin{definition}
    A joint distribution of prices $\jointpricedist \in \distover{\priceset^n}$ is a \emph{correlated equilibrium} if for each seller $i$, any realized price $p_i$ from the distribution is a best response conditional on $p_i$:
\begin{align*}
    p_i &\in \argmax_p \EExp{(p_i,p_{-i}) \sim \jointpricedist}{u_i(p,p_{-i}) \given p_i} \text{  for all $p_i$}. 
\end{align*}
\end{definition}

The regulator of pricing algorithms may not know the cost of the sellers.  \citet{Nekipelov2015} consider inferring both the costs and the external regrets of learning algorithms (Definition \ref{def:external-regret}).  They refer to the pairs of costs and external regrets that are consistent with the data to be the {\em rationalizable set}.  They show how to identify these rationalizable sets by assuming the pricing data contains counterfactual outcomes, i.e., what would have happened if a seller used a different price. We generalize this method to pricing data that does not contain counterfactual outcomes and the inference of calibrated regret.
\begin{definition}
    Given the historical price distributions $\{\pi_i^t\}^T_{t=1}$, prices $\{p_{-i}^t\}_{t=1}^T$ and demand functions $\{x_i^t(\cdot)\}_{t=1}^T$, a cost-regret pair $(c_i, \epsilon_i)$ for seller $i$ is \emph{rationalizable} if it satisfies
    \begin{equation}
    \forall \sigma,   \frac{1}{T}\sum_{t=1}^T\EExp{p_i^t \sim \pi_i^t}{(\sigma(p_i^t)-c_i)\cdot x_i^t(\sigma(p_i^t),p_{-i}^t) - (p_i^t-c_i)\cdot x_i^t(p_i^t,p_{-i}^t)} \leq \epsilon_i
    \end{equation}
  We define the rationalizable set $\mathcal{NR}_i\left(\{\pi_i^t\}^T_{t=1}, \{p_{-i}^t\}_{t=1}^T,\{x_i^t(\cdot)\}_{t=1}^T\right)$ as the set of all rationalizable pairs $(c_i, \epsilon)$. 
\end{definition}

Each point $(c_i, \epsilon_i)$ on the lower boundary of the rationalizable set gives the maximum expected calibrated regret $\epsilon_i$ of seller $i$ when she has cost $c_i$. The rationalizable set can be efficiently computed via the method provided in \citet{Nekipelov2015} with minimal assumptions.

\subsection{Collusive Equilibria in Repeated Games}

The setting of dynamic imperfect price competition is a \textit{repeated} game. On the other hand, correlated equilibria described previously is an equilibrium concept for a \textit{static} game (not repeated). Hindsight calibrated no-regret learning algorithms in the dynamic game, as we have seen, converge to this static equilibrium concept. The literature on repeated games, on the other hand, typically talks about dynamic equilibrium concepts, where an agent can explicitly condition on the actions of other agents in previous stages, perhaps to punish them for deviating from some prescribed strategy. We view such equilibria as collusive. The so-called ``folk theorems'' of repeated games describe outcomes that are possible as equilibria of the repeated game. 

\citet{benoit1984finitely} give a folk theorem for finitely repeated games. Stated in words: under weak conditions, any feasible and individually rational payoff of the static game can be approximated by the average payoff in a subgame-perfect equilibrium of a repeated game with a sufficiently long horizon.  A stable collusive outcome in a static pricing game is one where some players are best responding while other players are in a coalition and obtain higher individual payoffs than they would in an equilibrium that could result if they were to all be best responding.  In this outcome all best responding players are obtaining at least their individually rational payoffs and all colluding players are obtaining strictly more than their equilibrium payoffs which are at least their individually rational payoffs.  Thus, the folk theorem implies that equilibria in finitely repeated pricing games can approximate any stable collusive outcomes that exist.

\begin{corollary}
In a dynamic imperfect price competition game that is finitely repeated with a sufficiently long time horizon, any stable collusive outcome can be approximated by an equilibrium of the repeated game.
\end{corollary}

We have focused on hindsight calibrated no-regret learning algorithms that converge to correlated equilibria in the static pricing game. Another large family of learning algorithms that are natural to use for price competition is no-policy-regret learning algorithms. When a seller is learning how to price, it is natural for competitors to react to the seller's prices with their own pricing strategy. Policy regret algorithms compare their performance to the performance they could have achieved if they switched to a fixed policy and the others in the market responded to this switch.


\citet{arora2018policy} introduced the notion of a \textit{policy equilibrium} that corresponds to outcomes in games played by no-policy-regret learning algorithms. They showed that Policy equilibrium is a strictly larger class than correlated equilibrium. As correlated equilibrium corresponds to calibrated best response by each agent in each round, the policy equilibria that are not correlated equilibria are not best responding in each round. We view these outcomes as collusion. Hence, we view no-policy-regret learning algorithms as problematic for algorithmic pricing.

\section{Framework for Regulating Collusion}\label{sec:framework}

This section presents our definition of plausible non-collusion for sellers and an empirical framework for auditing it. In our model, the seller has a potentially private cost, which is static across rounds, and a potentially private signal that correlates with the demand (and possibly the competition, which might also correlate with the demand). Our framework is based on the following sufficient condition for non-collusion:

\begin{definition}
  It is non-collusive for a seller to approximately best respond to their
  competitive environment.
\end{definition}

While it is plausible that a seller who has not approximately best responded is not colluding, our framework will not be able to conclude that they have not colluded. The existence of algorithms that can easily satisfy the empirical definition we propose is evidence that it is permissible to hold sellers to such a standard. 

We may not know the seller's cost and/or the seller's beliefs on the competitive environment. In fact, these beliefs on the environment may be changing over time (though we assume that the sellers' costs are stationary). We will not require that the regulator knows anything about the seller's costs or beliefs. Instead, we will apply econometric principles of revealed preference and revealed information. If a seller is approximately best responding to their competitive environment, we can infer their cost and whether or not they are consistently using information that reveals what information they possess.  For this reason, our empirical notion of non-collusion is only \textit{plausible}, i.e., there exists a cost and belief that is \textit{consistent} with the data for which the seller has approximately best responded.

It is possible that sellers collude to act as though they have costs or information about the demand seemingly plausible to the regulator, but different from what they actually have. A regulator uninformed about the true cost and true information possessed by the sellers will not be able to detect such collusion. Our philosophy is that these possibilities exist already in the regulation of collusion absent algorithms, and our focus is on solving new challenges introduced by algorithms by essentially reducing them to the old challenges of regulating collusion.

Our definition of non-collusion is unilateral by definition. A seller can satisfy it regardless of the actions of other sellers. Specifically, it will not be important to explicitly model the detailed actions of other sellers, only the impact of those actions on the outcome of a seller.  A seller's outcome, given the actions of other sellers and buyers, is a function $x: \priceset \to [0,1]$ from their price $p \in \priceset$ to a quantity of goods sold at this price, a.k.a., a demand function. We will assume that the goods are normal goods, i.e., the demand function is monotonic where increasing price results in (weakly) decreasing allocation.

We first give a static definition of non-collusion that applies to a single round of pricing. We then generalize the definition to repeated pricing and allow for statistical learning.

\begin{definition}
A joint distribution on pairs of price and demand function $\jointpricedist \in \distover{\priceset \times (\priceset \to [0,1])}$ is in \emph{calibrated best-response} for a seller with cost $c$ if, conditioned on the seller's price $p$, $p$ is a best response:
\begin{align}
    p \in& \argmax_{p'} \EExp{(p,x) \sim \jointpricedist}{(p' - c)\, x(p')) \given p} & \forall\ p \in \priceset. 
\end{align}
\end{definition}

Calibrated best response captures what it means to be a good algorithm and allows the algorithm to use side information. Collusion is a potentially tacit agreement between sellers to keep prices higher than those in each seller's best interest, given the prices of the other sellers. On the other hand, best responding to the market and in particular what other sellers are doing is not collusion. Calibration allows side information. If the side information is useful, it manifests in distinct prices. The definition conditions the best response on the prices. In other words, calibration requires an internal consistency with respect to information that is revealed to be possessed in variation of prices. It is easy to observe that the calibrated best response is the unilateral version of correlated equilibrium. If all sellers' prices satisfy the calibrated best response then the joint distribution of prices is a correlated equilibrium.

While it might seem that allowing correlation is allowing collusion, we argue that, in fact, no reasonable definition of collusion can forbid correlation of prices. Specifically, non-collusion is inherently about best responding to market conditions. When consumer demand changes, the best response prices change. Consumer demand is something that all sellers should be measuring and it is correlated across sellers that are in price competition. Therefore, correlation must be allowed. Calibration is a minimal allowance of correlation and, in particular, it is agnostic to various potential sources of correlation and does not require that they be explicitly modeled.

\begin{definition}
    A joint distribution on pairs of price and demand $\jointpricedist$ is \emph{non-collusive} for a seller $i$ if $\jointpricedist$ satisfies the calibrated best response for $i$.
\end{definition}

In repeated environments, where sellers are learning about what prices are good, their prices might not be in the best response. However, as the learning proceeds, the distance from the best response should diminish. This property is captured by the following definition for the dynamic settings.

\begin{definition}
An infinite sequence of pairs of price and demand $\{(p^t,x^t)\}^t$ is \emph{calibrated vanishing regret} for a seller with cost $c$ if, the maximum average per-round benefit of deviation over the set of price remaps $\sigma: \priceset \to \priceset$,  up to a given round, approaches zero as the number of rounds goes to infinity: 
\begin{equation}
     \lim_{T\to \infty}\max_{\sigma}\frac{1}{T}\sum_{t=1}^Tu(\sigma(p^t),x^t)-u(p^t,x^t) = 0. \label{eqn:calibrated-vanishing-regret}
\end{equation}
where the payoff for a price $p$ on demand $x$ is $u(p,x) = (p - c)\, x(p)$
\end{definition}

Calibrated vanishing regret and calibrated best response are related in that:
\begin{itemize}
    \item If we draw a sequence of prices from a joint distribution that satisfies calibrated best response for the seller (and payoffs are bounded), then this sequence of prices will satisfy calibrated vanishing regret for the seller; and
    \item in the limit with the number of rounds, the uniform distribution on price-demand pairs (a.k.a. the empirical distribution) given by a sequence that satisfies calibrated vanishing regret for the seller approaches a distribution of prices in calibrated best responses for her.
\end{itemize}

These two properties give a unilateral version of an equivalence observed by  \citet{foster1997calibrated}: If the conditions hold for all sellers, then the empirical distribution of the price sequence approaches a correlated equilibrium.

Note that requiring calibration is important in our definition of non-collusion. The weaker notion of vanishing external regret (\Cref{def:external-regret}) does not require calibration, and it fails to rule out certain collusive behaviors when the sellers have private information about the demand. 

We demonstrate the problem with external regret with the numerical example discussed in Section \ref{sec:dynamic_price_comp}: Two sellers have cost $c_1=0.1$ and $c_2=0.2$ respectively and the buyer's valuations are i.i.d. uniform over $[0,1]$. However, Seller 1 now possesses private information. She can tell if an incoming buyer has a low valuation for both sellers, i.e., $v_1 \leq 0.5$ and $v_2 \leq 0.5$. Thus, she can post a different price for these buyers. Seller 1 can take advantage of this private information to collude with Seller 2 while still having non-positive external regret: Seller 2 posts a fixed price $p_2 = 0.66$. Seller 1 posts $p_1^L=0.3$ when she knows the buyer has a low valuation, and the same price $p_1=0.66$ as Seller 2 otherwise. Recall that external regret compares Seller 1's revenue against the best single price she could set. However, compared to a single price that always undercuts Seller 2, it is better off for her to get more revenue from low buyers and not compete against Seller 2 on high buyers.
On the other hand, in this example, Seller 1 does have positive calibrated regret. She can get a higher revenue by posting $p'_1=0.60$ whenever she posts $p_1=0.66$ under the current collusive strategy. In other words, the calibrated best response condition fails to hold, i.e., conditional on her posting $p_1=0.66$, $p_1=0.66$ is not a best response.

Our methods will not require the regulator to know the exact cost of a seller. It will be sufficient to know that the seller's cost is in a bounded range $\costrange$.  The regulator will assume the seller's regret is the minimum one that achieves costs in this range.

\begin{definition}
  An infinite sequence of price profiles of a seller is \emph{plausibly non-collusive} for cost range $\costrange$ if it satisfies calibrated vanishing regret for some cost $c \in
  \costrange$.
\end{definition}

There is a long literature that develops good learning algorithms for pricing with unknown demand, specifically by satisfying vanishing calibrated regret. Algorithms that do not satisfy vanishing calibrated regret are making mistakes in optimization that are apparent from the data. Given the information that the algorithms have which is revealed in the prices, they are not optimizing well enough that calibrated regret vanishes. We view this failure of optimization as a mistake, and algorithms that make this mistake as not good.

\begin{definition}
    An \emph{environment} is a process that generates the sequence of demands based on an algorithm's past decisions, i.e., a sequence of functions mapping a history of prices or distributions of prices, to a distribution of demand functions.
\end{definition}

A stochastic environment is an environment where the demand functions for each round are independent and identically distributed. An adversarial environment is an environment where the demand functions for each round are assumed to be generated in a way against the algorithm running in it.

\begin{definition}
    An algorithm is \emph{good} for a cost $c$ in an environment if it satisfies vanishing calibrated regret for cost $c$.
    \label{def:good}
\end{definition}


Calibrated vanishing regret cannot be directly observed in the data of a learning algorithm because (a) in practice, only data from a finite horizon can be observed, and (b) outcomes for counterfactual prices are not generally known.

With data observed from a finite horizon of length $T$, the methodology of property testing \citep{gol-10} can be used to check whether the expected calibrated regret of a seller at $T$ is below a threshold, which approximates the calibrated vanishing regret.
\begin{definition}
    The \emph{expected calibrated regret at time $T$} for a seller with cost $c$ against price remapping $\sigma$ is
\begin{equation}
    \Er{\sigma}{c} = \frac{1}{T}\sum_{t=1}^T\EExp{p^t \sim \pi^t}{u(\sigma(p^t),x^t)-u(p^t,x^t)}
\end{equation}
\end{definition}

\begin{definition} 
    The seller's \emph{plausible calibrated regret at time $T$} is $R^T_{\ast}= \min_{c\in\costrange}\max_{\sigma}\Er{\sigma}{c}$.
\end{definition}

Note that the seller's plausible calibrated regret $R^T_{\ast}$ is always smaller than her true calibrated regret $R^T = \max_{\sigma}\Er{\sigma}{c_0}$ when her cost $c_0 \in \costrange$.

While in round $t$, the seller uses price $p^t$ and obtains some utility for it, to test if her regret is low, we need counterfactual outcomes for other prices that could have been used, which we can only estimate based on the data.

We measure the statistical complexity of a low plausible calibrated regret test by the number of rounds $T$ that is sufficient to distinguish the two scenarios with high confidence:
\begin{itemize}
    \item the seller's true calibrated regret is below a given threshold (for sufficiently auditable algorithms); 
    \item the plausible calibrated regret of the seller is far above the given threshold.
\end{itemize}
This gives a two-sided bound while allowing for a failure to identify the low regret of algorithms that do not collect enough data to accurately make such a determination.

\begin{definition}[Sample complexity with auditability requirement]
    A low plausible calibrated regret test 
    has \emph{sample complexity $\Tthresh$ with auditability requirement $C$, confidence $1-\alpha$ and target regret level $\Rthresh$}, if $\Tthresh$ is the minimum $T$ such that
    \begin{itemize}
        \item if the seller's true calibrated regret $R^{T} \leq \Rthresh$ and the transcript satisfies the auditability requirement $C$, she passes the test with probability at least $1-\alpha$;
        \item if the plausible calibrated regret $\PlEr \geq 2\Rthresh$, the seller fails the test with probability at least $1-\alpha$.
    \end{itemize}
\end{definition}



Since counterfactual demand for other prices that could have been used can not be observed, an algorithm needs to keep additional data in the transcript to demonstrate that it has low regret.

Algorithms might not be designed to record such information. Our goal for auditing collusion is to allow any good algorithm to be used. Thus, we look for tests for which any algorithm can be retrofitted to collect the data so that, if their regret is low, they pass the test.

\begin{definition}
A low plausible calibrated regret test with auditability requirement $C$ is \emph{$(\eta_C,r_C)$-audit compatible} if the following holds: Any algorithm $A$ with expected calibrated regret at most $\Rthresh$ for any time horizon at least $T$ can be augmented to an algorithm that produces a transcript satisfying $C$ with expected calibrated regret at most $\Rthresh + r_C$ for a time horizon with expected length $(1+\eta_C)T$.
\end{definition}




To instantiate the above framework for auditing non-collusion, we must do the following:
\begin{itemize}
    \item define a low plausible calibrated regret test; 
    \item prove that the test has good sample complexity; 
    \item define a black-box transcription algorithm for converting any good learning algorithm into one that additionally produces an auditable transcript; and
\item prove that the test is audit-completable with a small loss and small additional time (by analyzing the transcription algorithm).
\end{itemize}
The next section completes these steps.

\section{Empirical Propensity Score Test}\label{sec:pscore}

In this section, we give one instantiation of our framework for auditing the collusion of a seller in dynamic imperfect price competition based on the propensity score estimator, which is a standard method in multi-armed bandit algorithms that have been developed for pricing.

Since we are focusing on one particular seller, as we did in the previous section, we drop the subscript $i$ from the notations for ease of reading and use $x^t(\cdot)$ to denote the demand determined by the environment at round $t$. We will also use \emph{regret} to refer to \emph{expected regret} for simplicity as we are not concerned with the realized regret.

\begin{definition}
The \emph{propensity score transcript} includes the sequences of
\begin{itemize}
    \item   distributions $\{\pi^t\}_{t=1}^T$ produced by the seller's algorithm,
    \item the actual prices posted $\{p^t\}_{t=1}^T$, and,
    \item the observed demand $\{x^t(p^t)\}_{t=1}^T$, i.e., the outcomes of posting price $p^t$ at round $t$, the seller experienced.
\end{itemize} \label{def:pscore-transcript}
\end{definition}

It is assumed that the price $p^t$ is actually drawn from the distribution $\pi^t$.  It is not hard for the seller to commit to doing so and convince the regulator with modern cryptography. 

With the price transcript described above, we define an estimated calibrated regret using the propensity score estimator for the unobserved probabilities of sale for counterfactual prices.
\begin{definition}
Given a price transcript, the \emph{propensity score estimator} for $x^t(\cdot)$ is
\begin{equation}
    \tilde{x}^t(p) =       
    \begin{cases}
        \frac{x^t(p^t)}{\pi^t(p^t)}  & \mbox{if $p = p^t$} \\
        0 & \mbox{otherwise}
    \end{cases}
  \end{equation}
\end{definition}
 The propensity score estimator weights the outcome of each observation inversely proportional to its rarity. Note that for any fixed $p$, $\tilde{x}^t(p)$ is an unbiased estimator for $x^t(p)$ as $\EExp{p\sim \pi}{\tilde{x}^t(p^t)} = x^t(p^t)$. We define the following \emph{estimated calibrated regret} for a seller with cost $c$ and 
against price remapping $\sigma : \priceset \to \priceset$:
\begin{align} \EstEr{\sigma}{c} &= \frac{1}{T}\sum_{t=1}^T\EExp{p^t \sim \pi^t}{\tilde{u}^t(\sigma(p^t)) - \tilde{u}^t(p^t)} 
\\&= \frac{1}{T}\sum_{t=1}^T\EExp{p^t \sim \pi^t}{(\sigma(p^t)-c)\cdot \tilde{x}^t(\sigma(p^t)) - (p^t-c)\cdot \tilde{x}^t(p^t)}. 
\end{align}
and the \emph{estimated calibrated regret} is $\max_\sigma\EstEr{\sigma}{c}$. The estimator estimates the true regret with the propensity score estimator for demand. 

We define \emph{minimum exploration probability} to quantify the exploration demonstrated by a transcript.
\begin{definition}
    The \emph{minimum exploration probability} of a transcript is
    \begin{equation}
        \underline{\pi}^T = \min_{p\in\mathcal{P},t\in\{1,\dots,T\}} \pi^t(p).
    \end{equation}
\end{definition}

To infer the cost of the seller, the regulator can compute the \emph{estimated rationalizable set} \`a la \citet{Nekipelov2015}: 
\begin{align*}\EstNR &= NR\left(\{\pi^t\}_{t=1}^T,\{p^t\}_{t=1}^T,\{\tilde{x}^t(\cdot)\}\right)\\ 
&= \biggl\{(c,\epsilon) : \forall \sigma, \frac{1}{T}\sum_{t=1}^T\EExp{p^t \sim \pi^t}{(\sigma(p^t)-c)\cdot \tilde{x}^t(\sigma(p^t)) \\
  &\phantom{= \biggl\{(c,\epsilon) : } -(p^t-c)\cdot \tilde{x}^t(p^t)}     \leq \epsilon \biggr\} \\
&= \left\{(c,\epsilon): \max_\sigma \EstEr{\sigma}{c} \leq \epsilon \right\}
\end{align*}
and find 
\begin{equation}
    \estplcost \in \argmin_{c \in \costrange}\left\{\epsilon : (c,\epsilon) \in \EstNR\right\} 
\end{equation}
as the \emph{estimated plausible cost} for the seller. $\estplcost$ is the cost with which the seller has the lowest estimated calibrated regret according to the data. ``Having cost $\estplcost$'' is a plausible explanation of the observed data that is most favorable in terms of estimated calibrated regret to the seller. 

To test if a seller's plausible calibrated regret $\PlEr \leq \Rthresh$ for target regret level $\Rthresh$, the regulator conducts the following test on a transcript defined in Definition \ref{def:pscore-transcript}.
\begin{definition}[Empirical propensity score test]
    Let the plausible calibrated regret estimator $\EstPlEr = \min_{\estplcost \in \costrange}\max_{\sigma}\EstEr{\sigma,\estplcost}$,
    and upper confidence bound $\EstUCB = \EstPlEr + \delta^T$ where the upper margin of error 
   \begin{equation*}
    \delta^T = \frac{k\pmax}{T}\sqrt{2\log\left(\frac{2k^2}{\alpha}\right)\cdot \sum_{s=1}^T\left(\frac{1}{\min_{p}\pi^s(p)}+1\right)^2},
    \end{equation*}
    with confidence $1-\alpha$, the number of price levels $k = |\priceset|$ and maximum possible price $\overline{p} = \max_{p}|\priceset|$.
    For target regret level $\Rthresh$, 
    \begin{itemize}
        \item pass if $\EstUCB\leq 2\Rthresh$; and
        \item fail, otherwise.
    \end{itemize}
\end{definition}

The accuracy of the propensity score regret estimator depends on how often the seller's algorithm explores. The estimation is accurate only when the algorithm explores often enough so that enough information is revealed. The upper margin of error term $\delta^T$ is added to the estimated calibrated regret in order to account for the error given the exploration of the seller's algorithm. This ensures that when the transcript fails to demonstrate that the algorithm producing it conducted enough exploration, it is hard for the seller to pass the test. Hence, a seller with high plausible calibrated regret can not pass the test for getting a low estimated plausible calibrated regret when the estimator is actually unreliable.

\begin{theorem}
   The empirical propensity score test has sample complexity
   \begin{equation}
       \Tthresh=\log\frac{2k^2}{\alpha}\cdot 2\left(\frac{k\pmax}{\epsilon}\right)^2\cdot \left(\frac{1}{\pimin} + 1\right)^2 
   \end{equation}
   with minimum exploration requirement $\underline{\pi}$, confidence $1-\alpha$, and  target regret level $\Rthresh$, 
   where $k = |\priceset|$ is the number of price levels, and $\overline{p} = \max_{p}|\priceset|$ is the maximum possible price. \label{thm:pscore-sample-complexity}
\end{theorem}

As discussed above, to be able to pass the empirical propensity score test, the seller's algorithm needs to explore often enough so that the transcript satisfies the minimum exploration requirement. 
The transcripts produced by an algorithm that does not explore often enough are not auditable using the empirical propensity score test even if the algorithms are actually non-collusive.

As long as an algorithm is robust enough in an environment, it can be modified to produce auditable transcripts in the same environment by mixing it with a small probability of uniform sampling of all the prices. Running the modification for a few more rounds will have roughly the same performance as the original one. An algorithm is robust in its operating environment if it regret is approximately preserved when the algorithm skips some rounds, as long as these rounds are randomly drawn independently from the algorithm and the environment.

\begin{definition}
    An algorithm $A$ is \emph{blackout robust in an environment} if the following holds: If running $A$ in the environment for any time horizon at least $T$ has regret at most $\Rthresh_T$,
    then for any time horizon $T' \geq T$,
    the regret of running $A$ on an independently selected subset of the $T'$ rounds with length at least $T$ is no greater than $\Rthresh_T$.
\end{definition}



\begin{theorem}
 Let $k = |\priceset|$ be the number of price levels and assume that the maximum possible price is normalized to $1$. Given any algorithm $A$, for minimum exploration requirement $\pimin$, consider the algorithm $\hat{A}$, which at each round $t$:
    \begin{itemize}
        \item w.p. $k\pimin$, output $p$ drawn uniformly from $\priceset$ 
        \item w.p. $1-k\pimin$, call $A$ with the inputs and output its output .\label{item:A-run}
   \end{itemize}
    Then,
    \begin{itemize}
        \item The distribution $\pi^t$ produced by algorithm $\hat{A}$ has minimum exploration probability at least $\pimin$.
        \item If $A$ is blackout robust in the environment and has regret at most $\Rthresh_T$ for any time horizon at least $T$, the regret of running $\hat{A}$ in the same environment until $A$ is called at least $T$ times is no greater than $\Rthresh_T + k \pimin$, and the expected number of rounds it takes is $T/(1-k \pimin)$.
    \end{itemize}  
    \label{thm:unif-exp-augment}
\end{theorem}

\begin{corollary}[Audit Compatibility]
b  Let $k=|\priceset|$ be the number of price levels and assume that the maximum possible price is normalized to $1$. The empirical propensity score test with minimum exploration requirement $\pimin$ is $(\eta,r)$-audit compatible for any blackout robust algorithms and environment where 
  $\eta = k\pimin /(1-k\pimin)$ and  $r= k\pimin$.
  \label{cor:audit-compatibility}
\end{corollary}

Good algorithms in various environments are automatically blackout robust, and thus can be modified to pass the empirical propensity score test.

\begin{lemma}
    Any no-regret algorithm in a stochastic environment is blackout robust in the environment.
    \label{lem:stoch-robust}
\end{lemma}

\begin{lemma}
    Any no-regret algorithm in an adversarial environment is blackout robust in the environment.
    \label{lem:advers-robust}
\end{lemma}

As we discussed in the Section \ref{sec:framework}, calibrated vanishing regret cannot be directly observed in the data of a learning algorithm from a finite horizon. The empirical propensity score test introduced in this section checks an approximation of the calibrated vanishing regret with data observed from a finite horizon of length $T$. 
We conclude this section by showing that the limiting behavior of the test as $T$ goes to infinity is consistent with checking whether plausible calibrated regret vanishes. 

The empirical propensity score test checks whether the upper confidence bound of the estimated plausible calibrated regret is below the target regret level, where the upper confidence bound of the estimated plausible calibrated regret is the sum of the plausible calibrated regret estimator $\EstPlEr=\max_{\sigma}\EstEr{\sigma}{\estplcost}$ and the error margin $\delta^T$.
To show the limiting behavior of the test, we establish the asymptotic consistency of the plausible calibrated regret estimator with the error margin: When the algorithm explores reasonably enough, if the true calibrated regret vanishes as $T$ goes to infinity, then so does the estimated plausible calibrated regret. On the other hand, when the regulator sets an increasing sequence of confidence levels that goes to one as $T$ goes to infinity and grows slowly enough, if the plausible calibrated regret does not vanish, then the upper confidence bound does not vanish. 
\begin{lemma}[Upper-bound consistency]
     Suppose the distributions $\{\pi^t\}_{t=1}^T$ produced by seller's algorithm satisfies
     \begin{equation}
     \sum_{t=1}^T\left(\frac{1}{\min_{p}\pi^t(p)}+1\right)^2 = o\left(T^2/\log(T)\right),\label{eqn:upper-consistent-cond}
     \end{equation} 
   then $\lim_{T\to \infty}\TrueEr\leq 0$ implies $\lim_{T\to \infty}\EstPlEr \leq 0$ almost surely.
   \label{lem:upper-consistency}
\end{lemma}

\begin{lemma}[Lower-bound consistency with error margin]
 Suppose the regulator chooses a vanishing sequence of $\{\alpha^T\}_T$ satisfying $\alpha^T = \Theta(T^{-2})$, if $\lim_{T\to \infty}\PlEr > 0$,
  then $\lim_{T\to \infty}\left(\EstPlEr + \delta^T\right) > 0$ almost surely.\label{lem:lower-consistency}
\end{lemma}

From the above two lemmas, we obtain the following theorem showing that the plausible calibrated regret estimator $\EstPlEr$ along with the error margin $\delta^T$ has the desirable limiting behaviors as $T$ goes to infinity. Algorithms satisfy calibrated vanishing regret and explore reasonably enough have vanishing upper confidence bound and can thus pass tests with target regret level vanishes as $T$ goes to infinity, while algorithms with non-vanishing plausible calibrated regret can not.  
\begin{theorem}
As $T\to \infty$, suppose the regulator chooses a vanishing sequence of $\{\alpha^T\}_T$ satisfying $\alpha^T = \Theta(T^{-2})$,
 \begin{itemize}
     \item if the seller's algorithm satisfies
     \begin{equation}
          \sum_{t=1}^T\left(\frac{1}{\min_{p}\pi^t(p)}+1\right)^2 = o\left(T^2/\log(T)\right),
     \end{equation}
     the upper confidence bound  $\EstUCB = \EstPlEr+\delta^T$ vanishes when the seller's true regret $\TrueEr$ vanishes almost surely;
     \item if seller's plausible calibrated regret $\PlEr$ does not vanish then the upper confidence bound $\EstUCB=\EstPlEr + \delta^T$ does not vanish almost surely.
 \end{itemize}
  \label{thm:estimated-reg-consistency}
\end{theorem}

\section{Conclusion}
In this work, we propose a definition for algorithmic non-collusion for pricing algorithms and a framework for empirically auditing non-collusion based on statistical tests on the data. Based on our framework, we give an instantiation with propensity score estimators and provide its statistical complexity.

The propensity score estimator for plausible maximum regret used in the empirical propensity score test introduced in Section \ref{sec:pscore} makes a few assumptions on the seller's algorithm. The seller's algorithm is required to either distribute some amount of probabilities on each action at every round or be robustly good in its operating environment. The accuracy and efficiency depend reversely on the magnitude of these probabilities. This raises the question of whether there are estimators without such restrictions.

One natural direction for future work is to find low plausible regret tests with lower statistical complexity or looser auditability requirements for transcripts.

Another interesting question is whether auditing non-collusion can be formulated as a continuing process without the regulator deciding a fixed time horizon $T$.

\begin{acks}
    This work started as a part of the IDEAL Spring 2021 Special Quarter on Data Science and Law and draws inspiration from the participants.
    We are grateful to James B. Speta for his valuable feedback on our initial draft. We thank Martino Banchio, Ilya R. Segal, and Xiao Wang for impactful discussions. We also appreciate the feedback and revision recommendations from the CSLaw’24 reviewers. We thank Chang Wang for pointing out typos and providing comments for the version of this paper published at CSLaw'24.
    The authors are partially supported by the National Science Foundation (CCF-1934931 and EECS-2216970). 
\end{acks}
\bibliography{ref.bib}
\pagebreak
\appendix
\section{Proofs}

\subsection{Theorem \ref{thm:pscore-sample-complexity}}
\begin{fact}[Azuma's Inequality\footnote{From 
 Michel Habib, Colin McDiarmid, Jorge Ramirez-Alfonsin, Bruce Reed, (1998), \textit{Probabilistic Methods
for Algorithmic Discrete Mathematics} Theorem 3.10. Proof see Theorem 13.4 of Michael Mitzenmacher and Eli Upfal (2017) \textit{Probability and Computating}} ]
    Given a sequence of random variables $\{Y_t\}_t$  and a filtration $\{\mathcal{F}_t\}_t$ such that $\mathbb{E}[Y_t\mid \mathcal{F}_{t-1}] = 0$. If there exists $\{d_t\}$ such that $|Y_t| \leq d_t$ for every $t$, then for any $\delta \geq 0$ and $T$:
    \begin{equation}
        \Pr\left[\sum_{t=1}^T Y_t\geq \delta \right] \leq \exp\left(-\frac{\delta^2}{2\sum_{t=1}^T d_t^2}\right),\,
        \Pr\left[\sum_{t=1}^T Y_t\leq -\delta \right] \leq \exp\left(-\frac{\delta^2}{2\sum_{t=1}^T d_t^2}\right). 
    \end{equation}
\end{fact}
\begin{lemma}
Let $c_0$ be seller's true cost, $\estplcost = \arg\min_{c}\max_{\sigma}\EstEr{\sigma}{c}$ be the estimated plausible cost, $\plcost = \arg\min_{c}\max_{\sigma}\Er{\sigma}{c}$ be the plausible cost, $k = |\priceset|$ be the number of price levels, $\min_p\pi^t(p) = \min_{p} \pi^t(p)$ be the minimum among the probabilities of posting each price level at round $s$ by the seller. We have
\begin{equation}
 \Pr[\max_{\sigma}\EstEr{\sigma}{\estplcost} \leq \max_{\sigma}\Er{\sigma}{c_0} + \delta ]
\geq 1-k^2\exp\left(-\frac{\delta^2}{2k^2\sum_{t=1}^T d_t^2}\right),
\end{equation}
\begin{equation}
\Pr[\max_{\sigma}\EstEr{\sigma}{\estplcost} \geq \max_{\sigma}\Er{\sigma}{\plcost} - \delta] \geq 1-2k^2 \exp\left(-\frac{\delta^2}{2k^2\sum_{t=1}^T d_t^2}\right)
\end{equation}
where
\begin{equation}
    d_t = \frac{1}{T}\left(\frac{1}{\min_p\pi^t(p)}+1\right)\pmax.
\end{equation}
\end{lemma}
\begin{proof}
Observe that for any fixed $c$, since remapping $p$ to $p'$ does not affect the payoff of $p'' \neq p'$, we have
\begin{equation}
    \max_{\sigma}\Er{\sigma}{c} = \sum_{p \in \priceset} \max_{p' \in \priceset} \sum_{t=1}^T\frac{1}{T} \Big(\pi^t(p)\left(u^t(p',c)-u^t(p,c)\right)\Big),\label{eqn:reg_decomp}
\end{equation}
and similarly
\begin{equation}
    \max_{\sigma}\EstEr{\sigma}{c} = \sum_{p \in \priceset} \max_{p' \in \priceset} \sum_{t=1}^T\frac{1}{T} \Big(\pi^t(p)\left(\tilde{u}^t(p',c)-\tilde{u}^t(p,c)\right)\Big).
\end{equation}
Let 
\begin{align}
    r^t_{p,p'}(c) = \frac{1}{T} \Big(\pi^t(p)\left(u^t(p',c)-u^t(p,c)\right)\Big), \tilde{r}^t_{p,p'}(c) = \frac{1}{T} \Big(\pi^t(p)\left(\tilde{u}^t(p',c)-\tilde{u}^t(p,c)\right)\Big),
\end{align}
\begin{equation}
    R^T_{p,p'}(c) = \sum_{t=1}^T r^t_{p,p'}(c), \tilde{R}^T_{p,p'}(c) = \sum_{t=1}^T \tilde{r}^t_{p,p'}(c),
\end{equation}
we have
\begin{equation}
\max_{\sigma}\Er{\sigma}{c} = \sum_{p\in\priceset}\max_{p' \in \priceset}R^T_{p,p'}(c),
\end{equation}
and
\begin{equation}
\max_{\sigma}\EstEr{\sigma}{c} = \sum_{p\in\priceset}\max_{p' \in \priceset}\tilde{R}^T_{p,p'}(c).
\end{equation}
We first show that for the deviation $\Delta R^T_{p,p'}(c) = \tilde{R}^T_{p,p'}(c) - R^T_{p,p'}(c)$ for each $p,p'$ is small with high probability using Azuma's Inequality.

Let $\Delta r^t_{p,p'}(c) = \tilde{r}^t_{p,p'}(c) - r^t_{p,p'}(c)$, we have
\begin{equation}
    \Delta R^T_{p,p'}(c) = \sum_{t=1}^T \Delta r^t_{p,p'}(c),
\end{equation} 
i.e., $\{\Delta r^t_{p,p'}(c)\}_t$ as the $\{Y_t\}$ in the formulation of Azuma's Inequality given above.

Let $\mathcal{F}_t$ be the  information available to the seller's algorithm up to $s$.  We first show that $\mathbb{E}[\Delta r_{p,p'}^t(c) \mid \mathcal{F}_{t-1}] = 0$, where. In fact,
\begin{equation}
  \Delta r^t_{p,p'}(c) = 
  \frac{1}{T}
  \cdot \pi^t(p)
  \Big(
  (p'-c)(\tilde{x}^t(p')-x^t(p'))-(p-c)(\tilde{x}^t(p)-x^t(p))
  \Big).
\end{equation}
For all $p \in \priceset$, by definition of 
$\tilde{x}^t(p)$, we have  $\mathbb{E}[ \tilde{x}^t(p) \mid \pi^t(p)] = \mathbb{E}[ \tilde{x}^t(p) \mid \mathcal{F}_{t-1}] = x^t(p)$ since $\pi^t(p)$ is determined by $\mathcal{F}_{t-1}$, as any algorithm can only use information available up to $t-1$ to compute its distribution of prices to post at round $s$. The same argument applies to $p'$. Therefore, by linearity of expectation, we have
$\mathbb{E}[\Delta r_{p,p'}^t(c) \mid \mathcal{F}_{t-1}] = 0$.

To apply Azuma's inequality, we now figure out the bound of the magnitude of each $\Delta r_{p,p'}^t(c)$. For $p=p'$ we have $\Delta r_{p,p'}^t(c)$, hence we assume below that $p \neq p'$. Let $\pmax = \max_{p\in \priceset}{p}$. By definition of $\widetilde{x}^t$, we have:
\begin{itemize}
    \item When $p'=p^t$,
    \begin{align}
        \Delta r_{p,p'}^t(c) = \frac{1}{T}
  \cdot \pi^t(p)
  \Big(
  (p'-c)\left(\frac{x^t(p')}{\pi^t(p')}-x^t(p')\right)-(p-c)(0-x^t(p))
  \Big).
    \end{align}
    Since for any $p\in \priceset$, $0 \leq \pi^t(p) \leq 1, 0 \leq x^t(p) \leq 1$ and $-\pmax \leq p-c \leq \pmax$,
    we have
    \begin{equation}
        |\Delta r_{p,p'}^t(c)| \leq \frac{1}{T}\left(\frac{1}{\pi^t(p')}+1\right)\pmax.
    \end{equation}
    \item When 
    $p=p^t$,
    \begin{align}
        \Delta r_{p,p'}^t(c) &= \frac{1}{T}
  \cdot \pi^t(p)
  \Big(
  (p'-c)\left(0-x^t(p')\right)-(p-c)\left(\frac{x^t(p)}{\pi^t(p)}-x^t(p)\right)
  \Big) \\
  &= \frac{1}{T}\left(
  (p'-c)(-\pi^t(p)x^t(p'))-(p-c)(x^t(p)-\pi^t(p)x^t(p))\right).
    \end{align}
    Since for any $p\in \priceset$, $0 \leq \pi^t(p) \leq 1, 0 \leq x^t(p) \leq 1$ and $-\pmax \leq p-c \leq \pmax$,
    we have
    \begin{equation}
        |\Delta r_{p,p'}^t(c)| \leq \frac{1}{T}\cdot 2\pmax.
    \end{equation}
    \item When $p'\neq p^t$ and $p\neq p^t$,
    \begin{equation}
         \Delta r_{p,p'}^t(c) = \frac{1}{T}
  \cdot \pi^t(p)
  \Big(
  (p'-c)\left(0-x^t(p')\right)-(p-c)\left(0-x^t(p)\right)
  \Big).
    \end{equation}
    Since for any $p\in \priceset$, $0 \leq \pi^t(p) \leq 1, 0 \leq x^t(p) \leq 1$ and $-\pmax \leq p-c \leq \pmax$,
    we have
    \begin{equation}
        |\Delta r_{p,p'}^t(c)| \leq \frac{1}{T}\cdot 2\pmax.
    \end{equation}
\end{itemize}
By the fact that $\pi^t(p')\leq 1$, we have
\begin{equation}
    \frac{1}{T}\left(\frac{1}{\pi^t(p')}+1\right)\pmax \geq \frac{1}{T}\cdot 2\pmax,
\end{equation}
therefore, in conclusion, 
\begin{equation}
 |\Delta r_{p,p'}^t(c)| \leq \frac{1}{T}\left(\frac{1}{\pi^t(p')}+1\right)\pmax.
\end{equation}

In Azuma's inequality, we can take $d_t$ to be \begin{equation}
    \frac{1}{T}\left(\frac{1}{\min_p\pi^t(p)}+1\right)\pmax, \label{eqn:appx-azuma-ds}
\end{equation}
which a uniform upper bound for $\Delta r_{p,p'}^t(c)$ over all $p,p'$.

\paragraph{Upper Tail} 
For any fixed $c$, we have,
\begin{align}
    &\Pr[\max_{\sigma}\EstEr{\sigma}{c} - \max_{\sigma}\Er{\sigma}{c} \geq \delta] \\  = &\Pr\left[\sum_{p \in \priceset} \max_{p'\in \priceset}\tilde{R}^T_{p,p'}(c) - \sum_{p \in \priceset} \max_{p''\in \priceset}R^T_{p,p''}(c) \geq \delta \right]
    \\ =
    &\Pr\left[\sum_{p \in \priceset} \max_{p'\in \priceset}\tilde{R}^T_{p,p'}(c) - \max_{p''\in \priceset} R^T_{p,p''}(c) \geq \delta \right]\\
    \leq 
    &\Pr\left[\exists p \in \priceset, \max_{p'\in \priceset}\tilde{R}^T_{p,p'}(c) - \max_{p''\in \priceset} R^T_{p,p''}(c) \geq \frac{\delta}{|\priceset|} \right]\\
    \leq 
    &\Pr\left[\exists p \in \priceset, \max_{p'\in \priceset}\left(\tilde{R}^T_{p,p'}(c) - R^T_{p,p'}(c)\right) \geq \frac{\delta}{|\priceset|} \right]\\
    =&\Pr\left[\exists p \in \priceset, \exists p'\in \priceset,\tilde{R}^T_{p,p'}(c) - R^T_{p,p'}(c) \geq \frac{\delta}{|\priceset|} \right]
\end{align}
The first inequality comes from the simple fact that at least one element of the sum must be no less than the average. Taking the union bound over all $p,p' \in \priceset$, we get
\begin{equation}
\Pr\left[\exists p \in \priceset, \exists p'\in \priceset,\tilde{R}^T_{p,p'}(c) - R^T_{p,p'}(c) \geq \frac{\delta}{|\priceset|} \right] \leq \sum_{p \in \priceset}\sum_{p'\in \priceset}\Pr\left[\tilde{R}^T_{p,p'}(c) - R^T_{p,p'}(c) \geq \frac{\delta}{|\priceset|}\right]. 
\end{equation}
Note that $|\priceset| = k$,  by Azuma's inequality with $d_t$ defined as above uniformly over all $p,p'$, we have
\begin{equation}
    \Pr\left[\tilde{R}^T_{p,p'}(c) - R^T_{p,p'}(c) \geq \frac{\delta}{|\priceset|}\right] = \Pr\left[\Delta R_{p,p'}^T(c) \geq \frac{\delta}{k}\right] \leq \exp\left(-\frac{\delta^2}{2k^2\sum_{t=1}^T d_t^2}\right).
\end{equation}
Therefore, for any fixed $c$,
\begin{equation}
    \Pr[\max_{\sigma}\EstEr{\sigma}{c} - \max_{\sigma}\Er{\sigma}{c} \geq \delta] \leq k^2\exp\left(-\frac{\delta^2}{2k^2\sum_{t=1}^T d_t^2}\right).\label{eqn:appx-deviation-uppertail}
\end{equation}

Now we consider the relationship between $\min_{c \in \costrange}\max_{\sigma}\EstEr{\sigma}{c}$ and $\max_{\sigma}\Er{\sigma}{c_0}$.
By the fact that the estimated plausible cost $\estplcost =\arg\min_{c}\max_{\sigma}\EstEr{\sigma}{c}$, we
have 
\begin{equation}
\begin{split}
&\Pr[\max_{\sigma}\EstEr{\sigma}{\estplcost} \leq \max_{\sigma}\Er{\sigma}{c_0} + \delta ]\\
\geq &\Pr[\max_{\sigma}\EstEr{\sigma}{c_0}
\leq \max_{\sigma}\Er{\sigma}{c_0} + \delta ]\\ \geq &1-k^2\exp\left(-\frac{\delta^2}{2k^2\sum_{t=1}^T d_t^2}\right) 
\end{split}.\label{eqn:appx-uppertail-bound}
\end{equation}

\paragraph{Lower Tail}
Let $\plcost = \arg\min_{c}\max_{\sigma} \Er{\sigma}{c}$ be the plausible cost of the seller. By definition of the estimated plausible cost $\estplcost$, we have $\max_{\sigma}\Er{\sigma}{\plcost} \leq  \max_{\sigma}\Er{\sigma}{\estplcost}$. Therefore,
\begin{equation}
    \max_{\sigma}\EstEr{\sigma}{\estplcost} - \max_{\sigma}\Er{\sigma}{\plcost} \geq \max_{\sigma}\EstEr{\sigma}{\estplcost} - \max_{\sigma}\Er{\sigma}{\estplcost},
\end{equation}
which implies
\begin{equation}
    \Pr[\max_{\sigma}\EstEr{\sigma}{\estplcost} - \max_{\sigma}\Er{\sigma}{\plcost} \leq - \delta] \leq  
    \Pr[\max_{\sigma}\EstEr{\sigma}{\estplcost} - \max_{\sigma}\Er{\sigma}{\estplcost} \leq -\delta].
\end{equation}

Note that since $\estplcost$ is a random variable, we can not simply treat it as a fixed $c$ and obtain a probability bound using the exact same argument as the upper tail for $\max_{\sigma}\EstEr{\sigma}{\estplcost}-\max_{\sigma}\Er{\sigma}{\estplcost}$. Instead, we consider the event across all fixed $c$.

\begin{align}
&\Pr[\max_{\sigma}\EstEr{\sigma}{\estplcost} - \max_{\sigma}\Er{\sigma}{\estplcost} \leq -\delta]\\
  = 
  &\Pr\left[\sum_{p \in \priceset} \max_{p'\in \priceset}\tilde{R}^T_{p,p'}(\estplcost) - \sum_{p \in \priceset} \max_{p''\in \priceset}R^T_{p,p''}(\estplcost) \leq -\delta \right]\\ 
  =
  &\Pr\left[\sum_{p \in \priceset} \max_{p''\in \priceset}R^T_{p,p''}(\estplcost) - \max_{p'\in \priceset} \tilde{R}^T_{p,p'}(\estplcost) \geq \delta \right]\\
  \leq
  &\Pr\left[\sum_{p \in \priceset} \max_{p''\in \priceset}R^T_{p,p''}(\estplcost) - \tilde{R}^T_{p,p''}(\estplcost) \geq \delta \right]\\
  \leq 
  &\Pr\left[\exists c \in \costrange, \sum_{p \in \priceset} \max_{p''\in \priceset}R^T_{p,p''}(c) - \tilde{R}^T_{p,p''}(c) \geq \delta \right]\\
 \leq &1-\Pr\left[\forall c \in \costrange, \sum_{p \in \priceset} \max_{p''\in \priceset}R^T_{p,p''}(c) -  \tilde{R}^T_{p,p''}(c) \leq \delta \right]. \label{eqn:appx-lowertail-breakpoint-1}
\end{align}
Note that with similar arguments for reasoning about the upper tail, we have
\begin{align}
    &\Pr\left[\forall c \in \costrange, \sum_{p \in \priceset} \max_{p''\in \priceset}R^T_{p,p''}(c) - \tilde{R}^T_{p,p''}(c) \leq \delta \right] \\
    \geq 
    &\Pr\left[\forall c \in \costrange, \forall p \in \priceset, \max_{p''\in \priceset}R^T_{p,p''}(c) - \tilde{R}^T_{p,p''}(c) \leq \frac{\delta}{|\priceset|} \right] \\
    =
    &\Pr\left[\forall c \in \costrange, \forall p \in \priceset, \forall p'' \in \priceset, R^T_{p,p''}(c) -  \tilde{R}^T_{p,p''}(c) \leq \frac{\delta}{|\priceset|} \right]\\
    =& \Pr\left[\forall c \in \costrange, \forall p \in \priceset, \forall p'' \in \priceset, \Delta R^T_{p,p''}(c) \ge -\frac{\delta}{|\priceset|} \right] \\
    \geq 
    &\Pr\left[ \forall p \in \priceset, \forall p'' \in \priceset, 
    \left(\Delta R_{p,p''}^T(\ubar{c}) \geq -\frac{\delta}{|\mathcal{P}|}  
    \right) \wedge
    \left( R_{p,p''}^T(\overline{c}) \geq -\frac{\delta}{|\mathcal{P}|}\right)
    \right]
\end{align}
The last inequality follows from the following observation: $\Delta R_{p,p''}^T(c)$ is linear in $c$, and hence its minimum over $c \in \costrange$ must be attained at either of the endpoints, hence we have
\begin{equation}
 \left[\Delta R_{p,p''}^T(\ubar{c}) \geq -\frac{\delta}{|\mathcal{P}|}\right] \cap \left[\Delta R_{p,p''}^T(\overline{c}) \geq -\frac{\delta}{|\mathcal{P}|}\right] 
 \subseteq
 \left[\forall c \in [\underline{c},\overline{c}], \Delta R_{p,p''}^T(c) \geq -\frac{\delta}{|\mathcal{P}|}\right]  .
\end{equation}
Therefore, continuing from (\ref{eqn:appx-lowertail-breakpoint-1}), we have
\begin{align}
    &1-\Pr\left[\forall c \in \costrange, \sum_{p \in \priceset} \max_{p''\in \priceset}R^T_{p,p''}(c) -  \tilde{R}^T_{p,p''}(c) \leq \delta \right] \\
    \leq 
    & 1- \Pr\left[\forall p \in \priceset, \forall p'' \in \priceset, 
    \left( \Delta R_{p,p''}^T(\ubar{c}) \geq -\frac{\delta}{|\mathcal{P}|} \right) \wedge
    \left( \Delta R_{p,p''}^T(\overline{c}) \geq -\frac{\delta}{|\mathcal{P}|} \right)
    \right] \\
    =
    & \Pr\left[\exists p \in \priceset, \exists p'' \in \priceset, 
    \left( \Delta R_{p,p''}^T(\ubar{c}) \leq -\frac{\delta}{|\mathcal{P}|}\right) \vee
    \left(\Delta R_{p,p''}^T(\overline{c}) \leq -\frac{\delta}{|\mathcal{P}|}\right)
    \right].
\end{align}
Taking the union bound, we have
\begin{align}
    &\Pr\left[\exists p \in \priceset, \exists p'' \in \priceset, 
    \left( \Delta R_{p,p''}^T(\ubar{c}) \leq -\frac{\delta}{|\mathcal{P}|} \right) \vee
    \left( \Delta R_{p,p''}^T(\overline{c})
    \leq -\frac{\delta}{|\mathcal{P}|}
    \right) \right] \\
    \leq &\sum_{p \in \priceset}\sum_{p'' \in \priceset}\left(\Pr \left[\Delta R_{p,p''}^T(\underline{c})
    \leq -\frac{\delta}{|\mathcal{P}|}\right] + \Pr\left[\Delta R_{p,p''}^T(\overline{c})
    \leq -\frac{\delta}{|\mathcal{P}|}\right]\right)
\end{align}
Note that $|\priceset|=k$, by Azuma's inequality with $d_t$ defined above uniformly over all $p,p''$, we have for any fixed $c$
\begin{equation}
    \Pr\left[\Delta R_{p,p''}^T(c) \leq -\frac{\delta}{k}\right] \leq \exp\left(-\frac{\delta^2}{2k^2\sum_{t=1}^T d_t^2}\right).
\end{equation}
Since $\ubar{c},\overline{c}$ are fixed value, we have
\begin{align}
    &\Pr\left[\exists p \in \priceset, \exists p'' \in \priceset, 
    \left( \Delta R_{p,p''}^T(\ubar{c}) \leq -\frac{\delta}{|\mathcal{P}|} \right) \vee
    \left( \Delta R_{p,p''}^T(\overline{c})
    \leq -\frac{\delta}{|\mathcal{P}|}
    \right) \right] \\
    \leq & k^2\cdot 2\exp\left(-\frac{\delta^2}{2k^2\sum_{t=1}^T d_t^2}\right),
\end{align}
and we conclude that
\begin{equation}
    \Pr[\max_{\sigma}\EstEr{\sigma}{\estplcost} \geq \max_{\sigma}\Er{\sigma}{\plcost} - \delta] \geq 1-2k^2 \exp\left(-\frac{\delta^2}{2k^2\sum_{t=1}^T d_t^2}\right) \label{eqn:appx-lowertail-bound}
\end{equation}
\end{proof}
\paragraph{Main Result of Theorem \ref{thm:pscore-sample-complexity}} Finally, we restate our theorem and give its proof: For a given confidence $1-\alpha$, target regret level $\epsilon$ and minimum exploration probability $\pimin$, when
\begin{equation}
    T > \log\frac{2k^2}{\alpha}\cdot 2\left(\frac{k\pmax}{\epsilon/2}\right)^2\cdot \left(\frac{1}{\pimin} + 1\right)^2,
\end{equation}
the following holds:
\begin{itemize}
    \item If $\max_\sigma\Er{\sigma}{c_0} \leq \epsilon$ and $\pimin^T=\min_{t,p}\pi^t(p) \geq \pimin$, then with probability at least $1-\alpha$,
    \begin{equation}
        \max_{\sigma}\EstEr{\sigma}{\estplcost} + \delta^T \leq 2\epsilon ;
    \end{equation}
    \item if $\max_\sigma\Er{\sigma}{\plcost} > 2\epsilon$, then with probability at least $1-\alpha$,
    \begin{equation}
        \max_{\sigma}\EstEr{\sigma}{\estplcost} + \delta^T > 2\epsilon.
    \end{equation}
    
\end{itemize}
where 
\begin{equation}
    \delta^T =  \sqrt{\log\frac{2k^2}{\alpha}} \cdot \frac{\sqrt{2}k\pmax}{T} \cdot \sqrt{\sum_{t=1}^T\left(\frac{1}{\min_p\pi^t(p)} + 1\right)^2}.
\end{equation}
with probability $1-\alpha$.
\begin{proof}
    Assume that
    \begin{equation}
        T >  \log\frac{2k^2}{\alpha}\cdot 2\left(\frac{k\pmax}{\epsilon}\right)^2\cdot \left(\frac{1}{\pimin} + 1\right)^2.
    \end{equation}
    \begin{itemize}
        \item If $\max_\sigma\Er{\sigma}{c_0} \leq \epsilon$ and $\min_{s,p}\pi^t(p) \geq \pimin$, from (\ref{eqn:appx-uppertail-bound}) we have, with probability at least $1-\alpha$ 
        \begin{equation}
            \max_\sigma\EstEr{\sigma}{\estplcost} \leq \max_\sigma\Er{\sigma}{c_0} + \delta^T \leq \epsilon
            +\delta^T, \label{eqn:appx-upper-w-cond}
        \end{equation}
        which implies that
        \begin{equation}
        \max_\sigma\EstEr{\sigma}{\estplcost} + \delta^T \leq \epsilon + 2\delta^T.
        \end{equation}
        Since $\min_{s,p}\pi^t(p) \geq \pimin$, 
        we also have $\delta^T \leq \epsilon/2$. Hence, with probability at least $1-\alpha$
        \begin{equation}
            \max_\sigma\EstEr{\sigma}{\estplcost} + \delta^T \leq 2\epsilon.
        \end{equation}
        \item If $\max_\sigma\Er{\sigma}{\plcost} > 2\epsilon$, from (\ref{eqn:appx-lowertail-bound}) we have, with probability at least $1-\alpha$,
        \begin{equation}
             \max_\sigma\EstEr{\sigma}{\estplcost} \geq \max_\sigma\Er{\sigma}{\plcost} - \delta^T > 2\epsilon
            -\delta^T,
        \end{equation}
        which implies that
        \begin{equation}
        \max_\sigma\EstEr{\sigma}{\estplcost} + \delta^T > 2\epsilon.
        \end{equation}
    \end{itemize}
    
\end{proof}
\paragraph{Skipped Algebraic Steps} We elaborate the algebraic steps in the skipped in the proof immediately above.

We first plug in
\begin{equation}
    d_t = \frac{1}{T}\left(\frac{1}{\min_p\pi^t(p)}+1\right)\pmax,
\end{equation}
from (\ref{eqn:appx-azuma-ds})
into (\ref{eqn:appx-uppertail-bound}) and (\ref{eqn:appx-lowertail-bound}).
Solving for $\delta$ fixing other parameters, i.e., for any given $T$, from (\ref{eqn:appx-uppertail-bound}) and (\ref{eqn:appx-lowertail-bound}) respectively, we get
\begin{equation}
    \delta^T_{u} = \sqrt{\log\frac{k^2}{\alpha}} \cdot \frac{\sqrt{2}k\pmax}{T} \cdot \sqrt{\sum_{t=1}^T\left(\frac{1}{\min_p\pi^t(p)} + 1\right)^2},
\end{equation}
\begin{equation}
    \delta^T_{l} = \sqrt{\log\frac{2k^2}{\alpha}} \cdot \frac{\sqrt{2}k\pmax}{T} \cdot \sqrt{\sum_{t=1}^T\left(\frac{1}{\min_p\pi^t(p)} + 1\right)^2}.
\end{equation}
$\delta^T$ is taken as $\max(\delta^T_u,\delta^T_l)$.

To see that $\max_\sigma\EstEr{\sigma}{\estplcost} \leq \max_\sigma\Er{\sigma}{c_0} + \delta^T$ and
$\delta^T \leq \epsilon/2$ holds
when $\min_{s,p}\pi^t(p) \geq \pimin$ and
\begin{equation}
     T \geq \log\frac{2k^2}{\alpha}\cdot 2\left(\frac{k\pmax}{\epsilon}\right)^2\cdot \left(\frac{1}{\pimin} + 1\right)^2 
\end{equation}
holds: Multiply both sides of the above inequality by $T$, we get
\begin{equation}
 T^2 \geq \log\frac{2k^2}{\alpha}\cdot 2\left(\frac{k\pmax}{\epsilon/2}\right)^2\cdot T\left(\frac{1}{\pimin} + 1\right)^2.
\end{equation}
Further, since for all $t=1,\dots,T$, $\min_p\pi^t(p) \geq \min_{s,p}\pi^t(p) = \pimin^T  \geq \pimin$, we have
\begin{equation}
    T\left(\frac{1}{\pimin} + 1\right)^2 \geq \sum_{t=1}^{T} \left(\frac{1}{\min_p\pi^t(p)} + 1\right)^2,
\end{equation}
hence
\begin{equation}
T^2 \geq \log\frac{2k^2}{\alpha}\cdot 2\left(\frac{k\pmax}{\epsilon/2}\right)^2\cdot \sum_{t=1}^{T} \left(\frac{1}{\min_p\pi^t(p)} + 1\right)^2,
\end{equation}
i.e.,
\begin{equation}
    T \geq \sqrt{\log\frac{2k^2}{\alpha}} \cdot \frac{\sqrt{2}k\pmax}{\epsilon/2} \cdot \sqrt{\sum_{t=1}^{T}\left(\frac{1}{\min_p\pi^t(p)} + 1\right)^2}.
\end{equation}
We get the desired results by plugging it into (\ref{eqn:appx-uppertail-bound}) and into $\delta^T$.

\vspace{1em}

\subsection{Theorem \ref{thm:unif-exp-augment}}
By definition of $\hat{A}$, at each round, with probability $k\pimin$ a price $p$ is chosen uniformly from $\priceset$. The resulted price distribution assigns at least $k\pimin / |\priceset|$ probability on each price level. Therefore we have the minimum exploration probability of $\hat{A}$ is at least $k\pimin / |\priceset| = \pimin$.

Let the set of rounds running $\hat{A}$ be $I = I_{A} \sqcup I_{U}$ where $I_{A}$ is the set of rounds where the algorithm $A$ is called and $I_{U}$ is the set of rounds where a price is uniformly chosen from $\priceset$. Let $T'=|I|$ and by assumption we have $|I_A| \geq T$. We have the regret running $\hat{A}$ on $I$
\begin{equation}
\max_{\sigma}\ExpRegret^I_{\hat{A}}(\sigma,c)  = \max_{\sigma} \mathbb{E}_{p^t \sim \pi^t}\left[\frac{1}{T'}\sum_{t \in I}{u^t(\sigma(p^t),c)-u^t(p^t,c)}\mid |I_A| \geq T\right]
\end{equation}

Let $R^t(\sigma,c) = u^t(\sigma(p^t)-c)-u^t(p^t,c)$ and the event $B = [|I_A| \geq T]$. By the law of iterated expectation, with the inner expectation fixing a realization of which rounds $A$ is called, we have

\begin{align}
    &\max_{\sigma}\ExpRegret^I_{\hat{A}}(\sigma,c)\\
    =&\max_{\sigma} \mathbb{E}_{p^t \sim \pi^t}\left[\mathbb{E}\left[\frac{1}{T'}\sum_{t \in I}R^t(\sigma,c)\mid I_A, B\right]\mid B\right]
    \\
    \leq &\max_{\sigma} \mathbb{E}_{p^t \sim \pi^t}\left[\mathbb{E}\left[\frac{1}{T'}\sum_{t \in I_U}R^t(\sigma,c)\mid I_A, B\right]\mid B\right] + \max_{\sigma} \mathbb{E}_{p^t \sim \pi^t}\left[\mathbb{E}\left[\frac{1}{T'}\sum_{t \in I_A}R^t(\sigma,c)\mid I_A, B\right]\mid B\right] \\
    \leq & \mathbb{E}_{p^t \sim \pi^t}\left[\max_{\sigma}\mathbb{E}\left[\frac{1}{T'}\sum_{t \in I_U}R^t(\sigma,c)\mid I_A, B\right]\mid B\right] +  \mathbb{E}_{p^t \sim \pi^t}\left[\max_{\sigma}\mathbb{E}\left[\frac{1}{T'}\sum_{t \in I_A}R^t(\sigma,c)\mid I_A, B\right]\mid B\right]
    \label{eqn:aug-reg-decomp}
\end{align}
Since the maximum possible price $\pmax$ is normalized to 1, we have
\begin{equation}
\max_{\sigma}\mathbb{E}\left[\sum_{t \in I_U}R^t(\sigma,c)\mid I_A, B\right] \leq |I_U|,
\end{equation}
which corresponding to the regret when posting prices uniformly sampled from $\priceset$. Combining the definition of $\hat{A}$ and $I_A$, we have the first term of (\ref{eqn:aug-reg-decomp})
is no greater than $\mathbb{E}[|I_U|/T' \mid B] = k\pimin$.

Note that by the assumption that $A$ is blackout robust, we have 
\begin{equation}  \max_{\sigma}\mathbb{E}\left[\frac{1}{|I_A|}\sum_{t \in I_A}R^t(\sigma,c)\mid I_A, B\right] \leq \Rthresh_T.
\end{equation}
Therefore we know that the second term of (\ref{eqn:aug-reg-decomp}) is also bounded by $\Rthresh_T$ as $T'\geq |I_A|$. 

Combining the bounds we conclude that $\max_{\sigma}\ExpRegret^I_{\hat{A}}(\sigma,c) \leq \Rthresh_T + k\pimin$.

By definition of $\hat{A}$, whether $A$ is called is independent across different rounds. Therefore $|T_U|$ satisfies a negative binomial distribution with success probability $1-\pimin k$ and success number $r=T$, and we have $\mathbb{E}[|T_U|] = T\pi k /(1-\pimin k)$. Therefore, the total expected number of rounds $\mathbb{E}[T'] = \mathbb{E}[|T_U|] + T = T/(1-k\pimin)$

\subsection{Corollary \ref{cor:audit-compatibility}}
\begin{proof}
    By combining the Theorem \ref{thm:unif-exp-augment}, Theorem \ref{thm:pscore-sample-complexity}, and the definition of audit compatibility.
\end{proof}

\subsection{Lemma \ref{lem:stoch-robust} and \ref{lem:advers-robust}}
Given a no regret algorithm $A$ in an environment, with (average) regret upper bound $\Rthresh_T=f(T)=o(1)$ for running for time horizon at least $T$ in the environment. When the environment is stochastic, i.e. $x^t(p) \sim D(p)$ for all $t$, any independently selected subset of $T'$ rounds with length at least $T$ still satisfies $x^t(p) \sim D(p)\, i.i.d.$. Therefore, the regret is at most $\Rthresh_T=f(T)$. When the environment is adversarial, the algorithm's regret bound should hold when it runs on any rounds with length at least $T$.

\subsection{Lemma \ref{lem:upper-consistency}}
For any $\delta >0$, Let $A_T$ be the event that $\EstEr{\sigma}{\estplcost}-\Er{\sigma}{c_0} < \delta$. To show that 
\begin{equation}
    \Pr\left[\liminf_{s \to \infty} A_T\right] = \Pr\left[\EstEr{\sigma}{\estplcost}-\Er{\sigma}{c_0} < \delta\, (e.v.)\right] = 1,
\end{equation}
it suffices to show that 
\begin{equation}
\Pr\left[\EstEr{\sigma}{\estplcost}-\Er{\sigma}{c_0} \geq \delta\, (i.o.)\right] = \Pr\left[\left[\EstEr{\sigma}{\estplcost}-\Er{\sigma}{c_0} < \delta\, (e.v.)\right]^C\right] = 0.
\end{equation}
By First Borel-Cantelli Lemma\footnote{cf. 2.7 of
David Williams (1991), \textit{Probability with Martingales}} and from (\ref{eqn:appx-uppertail-bound}), it suffices to have
\begin{equation}
    \sum_{T=1}^\infty \Pr\left[\EstEr{\sigma}{\estplcost}-\Er{\sigma}{c_0} \geq \delta \right] = \sum _{T=1}^{\infty}k^2\exp\left(-\frac{\delta^2}{2k^2\sum_{t=1}^T d_t^2}\right) < \infty.
\end{equation}
It suffices to have
\begin{equation}
   k^2\exp\left(-\frac{\delta^2}{2k^2\sum_{t=1}^T d_t^2}\right) = k^2 \exp\left(-\frac{\delta^2}{2k^2}\cdot \frac{T^2}{\sum_{t=1}^T \left(\frac{1}{\min_p\pi^t(p)}+1\right)^2}\right) = O\left(\frac{1}{T^2}\right). \label{eqn:bc-summation-cond1}
\end{equation}
as $\sum_{T=1}^\infty T^{-2} < \infty$.
It suffices to have
\begin{equation}
   D(T)=\frac{T^2}{\sum_{t=1}^T \left(\frac{1}{\min_p\pi^t(p)}+1\right)^2} = \omega(\log T), \label{eqn:bc-summation-cond2}
\end{equation}
which is equivalent to
\begin{equation}
     \sum_{t=1}^T\left(\frac{1}{\min_{p}\pi^t(p)}+1\right)^2 = o\left(T^2/\log(T)\right).
\end{equation}
\paragraph{Elaboration of (\ref{eqn:bc-summation-cond1}) implied by (\ref{eqn:bc-summation-cond2})}

In fact, by definition
\begin{equation}
    D(T) = \omega(\log T) \Leftrightarrow
    \forall c>0, \exists T_0,\mbox{ s.t. } \forall T>T_0, c\cdot D(T) > \log T.
\end{equation}
Take $c= \delta^2/4k^2$, we have for some $T_0$, for all $T>T_0$,
\begin{equation}
    \frac{\delta^2}{2k^2}\cdot D(T) > 2 \log T,
\end{equation}
which implies that for all $T>T_0$,
\begin{equation}
    k^2\exp\left(-\frac{\delta^2}{2k^2}\cdot D(T)\right) < \frac{1}{T^2},
\end{equation}
which by definition implies (\ref{eqn:bc-summation-cond1}).

\subsection{Lemma \ref{lem:lower-consistency}}
Given a sequence $\{\delta^T\}_T$, let $A_T$ be the event that 
$\EstEr{\sigma}{\estplcost} - \Er{\sigma}{\plcost} > - \delta^T$. To show that
\begin{equation}
\Pr\left[\liminf_{s \to \infty} A_T\right] = \Pr\left[\EstEr{\sigma}{\estplcost} -\Er{\sigma}{\plcost} > - \delta^T (e.v.)\right] = 1,
\end{equation} we follow the same steps as the proof of Lemma (\ref{lem:upper-consistency}) and conclude that it suffices to have
\begin{equation}
    \sum_{T=1}^\infty \Pr\left[\EstEr{\sigma}{\estplcost} - \Er{\sigma}{\plcost} \leq - \delta^T \right] = \sum _{T=1}^{\infty}2k^2\exp\left(-\frac{(\delta^T)^2}{2k^2\sum_{t=1}^T d_t^2}\right) < \infty.
\end{equation}
Plugging in the definition of $\delta^T$, the summand becomes $\alpha^T$. Hence, it suffices to have $\alpha^T = O(T^{-2})$.
\subsection{Theorem \ref{thm:estimated-reg-consistency}}
Note that when 
\begin{equation}
     \sum_{t=1}^T\left(\frac{1}{\min_{p}\pi^t(p)}+1\right)^2 = o\left(T^2/\log(T)\right),
\end{equation}
and $\alpha^T = \Theta(T^{-2})$, we have
$\delta^T = o(1)$. Combining with the Lemma (\ref{lem:upper-consistency}) we obtain the statement of the first bullet point.

The second bullet point directly comes from Lemma (\ref{lem:lower-consistency}).

\end{document}